\documentclass{article}

\usepackage{amsmath,amssymb}
\usepackage{graphicx}
\usepackage{wrapfig}
\usepackage{here}
\usepackage{amsthm}

\newtheorem{dfn}{Definition}
\newtheorem{lem}{Lemma}
\newtheorem{thm}{Theorem}
\newtheorem{exm}{Example}
\newtheorem{rmk}{Remark}
\newtheorem{prop}{Proposition}

\title{\bf \mbox{\boldmath{$D$}}-dimensional cellular automata provide Salem's singular function \mbox{\boldmath{$L_{\alpha}$}} with \mbox{\boldmath{$\alpha = 1/(2 D +1)$}} and \mbox{\boldmath{$1/(2^D+1)$}}}
\author{Akane Kawaharada\footnote{E-mail: aka@kyokyo-u.ac.jp, Postal address: 1, Fujinomoricho, Fukakusa, Fushimi-ku, Kyoto-shi, Kyoto, 612-8522, Japan} \vspace{2mm}\\
Department of Mathematics, Kyoto University of Education}

\begin{document}

\maketitle

\begin{abstract}
Salem's singular function is strictly increasing, continuous, and has a derivative equal to zero almost everywhere in $[0,1]$; it is also known as de Rham's singular function or Lebesgue's singular function.
The parameter of Salem's singular function $L_{\alpha}$ is $\alpha \in (0, 1)$ and $\alpha \neq 1/2$.
Our previous studies have shown that for some cases of which the limit set of spatio-temporal pattern of a cellular automaton (CA) is fractal, Salem's singular function with $\alpha = 1/3$, $1/4$, or $1/5$ is given by projecting the pattern onto the time axis.
However, it remained unclear whether there exists a CA that gives Salem's singular function with a parameter $\alpha$ equal to the multiplicative inverse of an integer greater than $5$.
In this paper, we construct CAs giving Salem's singular function with $\alpha= 1/(2D+1)$ and $\alpha = 1/(2^D+1)$ for each dimension $D \geq 1$.
This implies that there exist CAs that give Salem's function with a parameter $\alpha$ equal to the multiplicative inverse of any integer greater than or equal to $3$.
We also present the results of numerical experiments showing that for $D \leq 5$, the functions given by $D$-dimensional linear symmetric $2$-state radius-$1$ CAs other than the above two types cannot be Salem's function with $\alpha = 1/M$ for $M \in {\mathbb Z}_{\geq 3}$.
In addition to the square lattice, the triangular and hexagonal lattices can be considered as regular lattices in the two-dimensional plane, and we also discuss functions obtained from CAs on these lattices.
%
%
\end{abstract}

\hspace{2.5mm} 

\section{Introduction}
\label{sec:intro}

Cellular automata (CAs) are discrete mathematical models that can generate complex behaviors even from simple rules.
CAs are also useful as fractal generators, because orbits from the single site seed are often self-similar or partially self-similar. 
Such orbits have been studied by many researchers, 
for example, \cite{willson1984, culik1989, takahashi1992, haeseler1993}, and we can see many fractal patterns by CAs in \cite{wolfram2002}.
In general, fractals are often characterized by a fractal dimension. 
However, because the dimensionality of fractals is given by a single value, the same value may be given to different fractals.
Therefore, we propose a method for classifying fractals generated by these CAs using a function that describes the dynamics of the number of nonzero states of the spatio-temporal pattern \cite{kawa2022P}.

In this study, we characterize fractals in more detail by projecting the spatio-temporal pattern of a CA onto the time axis to yield a function when the limit of the spatio-temporal pattern is a fractal.
Especially for fractals generated by high-dimensional CAs, this function makes it possible to reduce the dimensionality of the fractal and analyze it.
High-dimensional fractals have a wide range of applications, for example, in image analysis \cite{jacquin1992, wohlberg1999} and fractal blinding \cite{sakai2012} in engineering, models of blood vessels and intestinal structure in biology \cite{family1989, roy2004}, as well as financial models \cite{mandelbrot2005, fernandez2019} in finance, etc.
Visualizing high-dimensional patterns accurately is difficult and costly.
Therefore, a new means of high-dimensional fractal analysis using a function that preserves more information about the original fractal is needed, as opposed to a fractal dimension that characterizes the pattern with a single numerical value.
Furthermore, in this study, we aim to characterize and classify `pathological functions'.
Pathological functions are those with deviant, irregular, or counterintuitive properties, although the term is not rigorously defined mathematically.
Some examples are well-known, such as the Cantor function, which is continuous but not absolutely continuous, the Takagi function, which is continuous but non-differentiable everywhere, and Thomae's function, which is discontinuous at rational points and continuous at irrational points \cite{cantor1884, takagi2014, thomae1875, kharazishvili2017}.
A singular function may be described as pathological, because they are non-constant, continuous, differentiable almost everywhere, and their derivative is zero.
Our previous studies have shown that functions derived from the spatio-temporal patterns of CAs can be pathological functions, such as singular functions \cite{kawanami2019, kawanami2020, kawa2022} or discontinuous Riemann-integrable functions \cite{kawa202103}.
In the future, by considering functions for many types of CAs, we expect it to become possible to characterize and classify pathological functions in association with fractals.

The spatio-temporal patterns of linear CAs are self-similar or partially self-similar, and the number of nonzero states generating the pattern is countable.
The authors counted the number of nonzero states at each time step and the cumulative number of nonzero states in the spatio-temporal pattern from the initial configuration with only a single positive state for several elementary CAs \cite{kawanami2020}.
Based on the result of counting the number of nonzero states, the dynamics of the number of states is normalized, and the function is given by taking the limit. 
This corresponds to projecting the limit set generated by each CA onto the time axis and expressing it as a function of a single variable.
In \cite{kawa2022}, the author showed that for a one-dimensional elementary CA Rule $150$, the resulting function is a new kind of strictly increasing singular function. 
In \cite{kawa2022P}, the results obtained for one-dimensional and two-dimensional elementary CAs that generate symmetric patterns are summarized. 
The functions obtained from these spatio-temporal patterns were found to be strictly increasing singular functions, and we obtained sufficient conditions for a function to be a strictly increasing singular function by extracting their common properties.
In particular, for the one-dimensional elementary CA Rule $90$ and a two-dimensional elementary CA, the functions obtained from their spatio-temporal patterns are self-affine functions called Salem's singular function $L_{\alpha}$ with parameters $\alpha= 1/3$ and $\alpha=1/ 5$, respectively. 
Including nonlinear CAs, we found the function for a two-dimensional nonlinear elementary CA to be $L_{1/4}$ \cite{kawanami2020}.

Salem's singular function is strictly increasing and continuous, and its derivative is zero almost everywhere in the interval $[0,1]$. 
Salem's singular function is also known as de Rham's singular function or Lebesgue's singular function \cite{lomnicki1934, salem1943, derham1957, takacs1978, yhk1997}.
Its mathematical properties and relationship to natural phenomena have been studied by many researchers.
The relationship between this function and the Takagi function is shown in \cite{hatayama1984}, and its relationship with complex dynamical systems is discussed in \cite{sumi2009, sumi2010}.
The set of points where the derivative of this function is zero and infinity is discussed in \cite{kawamura2011}.
Salem's singular function has been studied in connection with various other fields, including in research on gambling \cite{billingsley1983, stroock2000, neidinger2016}, the digital sum problem \cite{okada1995, kruppel2009, kruppel2011}, and physics-related studies \cite{takayasu1984, tasaki1993, bernstein1993}.
The parameter of Salem's singular function $L_{\alpha}$ is $\alpha \in (0, 1)$, where $\alpha \neq 1/2$.
Our previous studies have shown that the functions given by the spatio-temporal patterns of CAs are Salem's singular function $\alpha = 1/3$, $1/4$, and $1/5$, but whether there exists a CA that gives Salem's singular function $L_{1/M}$ for $M \in {\mathbb Z}_{> 5}$ remained unclear.
In the present work, we show that there exists a CA that gives a function $L_{1/M}$ for any odd number $M$ greater than or equal to $3$.
The CAs considered in this paper are equipped with two states $\{0 ,1\}$, on a $D$-dimensional square lattice ${\mathbb Z}^D$ having linear transition rules with radius $1$, and holding some symmetries.
Here, we adopt the simplest state set and the shortest radius of the nontrivial CAs, but their dimension $D$ takes any positive integer.
For each dimension $D$, we construct a CA that gives Salem's singular function with the parameter $\alpha= 1/(2D+1)$.
Furthermore, we also construct a CA that gives the singular function with $\alpha = 1/(2^D+1)$.
We present the results of numerical experiments showing that the function $L_{1/M}$ for $M \in {\mathbb Z}_{\geq 3}$ cannot be given by a $D$-dimensional linear symmetric $2$-state radius-$1$ CA other than the above two types for $2 \leq D \leq 5$.
In addition to the square lattice, the triangular and hexagonal lattices are also considered as regular lattices in the two-dimensional plane, and we also discuss functions obtained from CAs on these lattices.

The remainder of this paper is organized as follows. 
Section~\ref{sec:pre} presents some preliminaries on CAs and a singular function. 
For $D$-dimensional linear symmetric $2$-state radius-$1$ CAs $F_D$ and $G_D$, Section~\ref{sec:main} reports our main results that the functions representing the number of nonzero states in the spatio-temporal patterns are Salem's singular function with the parameter $\alpha = 1/(2D+1)$ and $1/(2^D+1)$. 
In Section~\ref{sec:other}, we provide some numerical results on $D$-dimensional linear symmetric $2$-state radius-$1$ CAs except $F_D$ and $G_D$, and CAs on the triangular and hexagonal lattices in the two-dimensional plane.
Finally, Section~\ref{sec:cr} discusses our findings and highlights some possible avenues for future research.

\section{Preliminaries}
\label{sec:pre}

Here, we define CAs and a singular function. 

\begin{dfn}
Let $\{0, 1\}$ be a binary state set.
We define a $D$-dimensional configuration space $\{0, 1\}^{{\mathbb Z}^D}$ for $D \in {\mathbb Z}_{>0}$, and each element $u \in \{0, 1\}^{{\mathbb Z}^D}$ is called a configuration.
A $D$-dimensional cellular automaton $(\{0, 1\}^{{\mathbb Z}^D}, T)$ for ${\textit{\textbf i}} \in {\mathbb Z}^D$ is given by
\begin{align}
(T u)_{\textit{\textbf i}} &= f ( u_{{\textit{\textbf i}}+v_1}, u_{{\textit{\textbf i}}+v_2}, \ldots, u_{{\textit{\textbf i}}+v_l}),
\end{align}
where $v_1, v_2, \ldots, v_l \in {\mathbb Z}^D$ are finite pairwise distinct vectors and $f : \{0, 1\}^l \to \{0, 1\}$ for $l \in {\mathbb Z}_{>0}$ is a local rule.
\end{dfn}

Next, we provide some properties of CAs.

\begin{dfn}
\begin{enumerate}
\item The radius of a CA is given by $\max_{1 \leq j \leq l} |v_j|$. 
We refer to a CA as radius-$1$ if the local rule satisfies $\max_{1 \leq j \leq l} |v_j| =1$.
\item We refer to a CA as linear if the local rule is given by 
\begin{align}
f ( u_{{\textit{\textbf i}}+v_1}, u_{{\textit{\textbf i}}+v_2},  \ldots, u_{{\textit{\textbf i}}+v_l}) = \sum_{j=1}^l a_j u_{{\textit{\textbf i}}+v_j} \pmod 2
\end{align}
for $a_j \in \{0, 1\}$.
Then, a radius-$1$ CA is linear if the local rule is given by 
\begin{align}
& f ( u_{{\textit{\textbf i}}+(-1, \ldots, -1, -1)}, u_{{\textit{\textbf i}}+(-1, \ldots, -1, 0)}, \ldots, u_{{\textit{\textbf i}}+(1, \ldots, 1, 1)}) \nonumber \\
&= \sum_{e_1=-1}^1 \sum_{e_2=-1}^1 \cdots \sum_{e_D=-1}^1 a_{e_1, e_2, \ldots, e_D} u_{i_1+e_1, i_2+e_2, \ldots, i_D+e_D} \pmod 2.
\end{align}
\item We refer to a linear radius-$1$ CA as symmetric if the local rule holds both reflection symmetry for each axis and discrete rotational symmetry of the $4$th order for each plane given by two axes, i.e., 
\begin{align}
& a_{e_1, e_2, \ldots, e_D} = a_{-e_1, e_2, \ldots, e_D} = a_{e_1, -e_2, e_3, \ldots, e_D} = \cdots = a_{e_1, \ldots, -e_D} \mbox{ and} \label{eq:axis}\\ 
& a_{e_1, e_2, e_3, \ldots, e_D} = a_{e_2, -e_1, e_3, \ldots, e_D} = a_{e_1, e_3, -e_2, \ldots, e_D} = \cdots = a_{e_1, \ldots, e_{D-2}, e_D, -e_{D-1}} \label{eq:rotate}
\end{align}
for $e_1, e_2, \ldots, e_D \in \{-1, 0, 1\}$.
\end{enumerate}
\end{dfn}

Let ${\mathcal C}(D)$ be the set of $D$-dimensional linear symmetric $2$-state radius-$1$ CAs.

\begin{prop}
\label{prop:2d+1}
The number of CAs belonging to ${\mathcal C}(D)$ is $2^{D+1}$ for each $D \in {\mathbb Z}_{>0}$.
\end{prop}

\begin{proof}
For a $D$-dimensional radius-$1$ CA, the local rule depends on $3^D$ states of neighboring cells.
Because of axisymmetries of the local rule by Equation~\eqref{eq:axis}, 
it suffices to consider only the following two cases, $e_j =0$ and $e_j=1$ of the coefficient $a_{e_1, e_2, \ldots, e_D}$ for each $1 \leq j \leq D$.
Hence, we have $2^D$ types of the coefficients $a_{e_1, e_2, \ldots, e_D}$s.
In addition, because of rotational symmetries of the local rule by Equation~\eqref{eq:rotate},
the coefficients $a_{e_1, e_2, \ldots, e_D}$s are the same 
when $(e_j, e_{j+1}) = (0, 1)$ and $(1, 0)$, 
when $(e_j, e_{j+1}) = (1, 0)$ and $(0, -1)$, or 
when $(e_j, e_{j+1}) = (1, 1)$ and $(1, -1)$, for each $1 \leq j < D$.
Considering both the axisymmetries and the rotational symmetries, 
the coefficients $a_{e_1, e_2, \ldots, e_D}$s are the same when $\sum_{i=1}^D |e_i|$ matches.
Therefore, we obtain $D+1$ types of $a_{e_1, e_2, \ldots, e_D}$s.
Because each cell has $0$ or $1$ as a state for $2$-state CAs, we have $2^{D+1}$ $D$-dimensional linear symmetric $2$-state radius-$1$ CAs.
\end{proof}

\begin{rmk}
\label{rmk:triv}
For ${\mathcal C}(D)$, two of the CAs have only trivial orbits because $(T^n u)_{\textit{\textbf i}} = 0$ and $(T^n u)_{\textit{\textbf i}}=u$ for any time step $n \in {\mathbb Z}_{>0}$.
Thus, we consider the other $2^{D+1}-2$ $D$-dimensional linear symmetric $2$-state radius-$1$ CAs for each $D \in {\mathbb Z}_{>0}$. 
\end{rmk}

We define the configuration $u_o \in \{0,1\}^{{\mathbb Z}^D}$ by
\begin{align}
\label{eq:sss}
(u_o)_{{\textit{\textbf i}}} = \left\{
\begin{array}{l l}
1 & \mbox{if ${\textit{\textbf i}} = {\bf 0} \ ( = (0, 0, \ldots, 0))$},\\
0 & \mbox{if ${\textit{\textbf i}} \in {{\mathbb Z}^D} \backslash \{ \bf 0 \}$},
\end{array}
\right. 
\end{align} 
for ${\textit{\textbf i}} \in {\mathbb Z}^D$. 
We refer to $u_o$ as the single site seed.
In this study, we investigate the orbits of CAs from the single site seed $u_o$ as an initial configuration.
For a $D$-dimensional CA $(\{0, 1\}^{{\mathbb Z}^D}, T)$, let $num_T (n)$ be the number of $1$-states in a spatial pattern $T^n u_o$, and $cum_T (n)$ be the number of $1$-states in a spatio-temporal pattern $\{T^m u_o\}_{m=0}^n$.
Then, 
\begin{align}
num_T (n) = \sum_{{\textit{\textbf i}} \in {\mathbb Z}^D} (T^n u_o)_{\textit{\textbf i}}, \
cum_T (n) = \sum_{m = 0}^n \sum_{{\textit{\textbf i}} \in {\mathbb Z}^D} (T^m u_o)_{\textit{\textbf i}}.
\end{align}

%
%
%

Below, we define a singular function related to CAs.

\begin{dfn}[\cite{derham1957, yhk1997}]
\label{def:lb}
Let $\alpha$ be a parameter such that $0 < \alpha < 1$ and $\alpha \neq 1/2$.
Salem's singular function $L_{\alpha}:[0,1] \to [0,1]$ is defined by
\begin{align}
\label{eq:lb}
L_{\alpha} (x):=
\left\{
\begin{array}{l l}
\alpha L_{\alpha} (2 x) & \ (0 \leq x < 1/2),\\
(1-\alpha) L_{\alpha}(2 x-1) + \alpha & \ (1/2 \leq x \leq 1).
\end{array}
\right. 
\end{align}
\end{dfn}
The functional equation~\eqref{eq:lb} has a unique continuous solution on $[0,1]$.
The resulting function $L_{\alpha}$ is strictly increasing, continuous, and has a derivative of zero almost everywhere. 
Figure~\ref{fig:Lx} shows the graphs of $L_{\alpha}$ for $\alpha = 1/3$, $1/5$, and $1/9$. 

\begin{figure}[h]
\begin{minipage}[b]{0.32\linewidth}
\centering
\includegraphics[width=1.\linewidth]{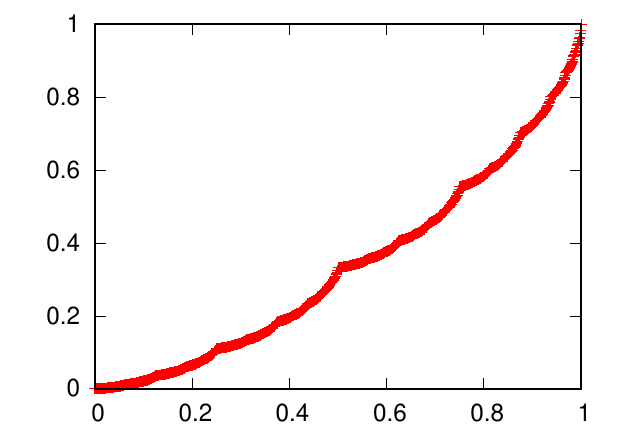}\\
(a) $\alpha = 1/3$
\end{minipage}
\begin{minipage}[b]{0.32\linewidth}
\centering
\includegraphics[width=1.\linewidth]{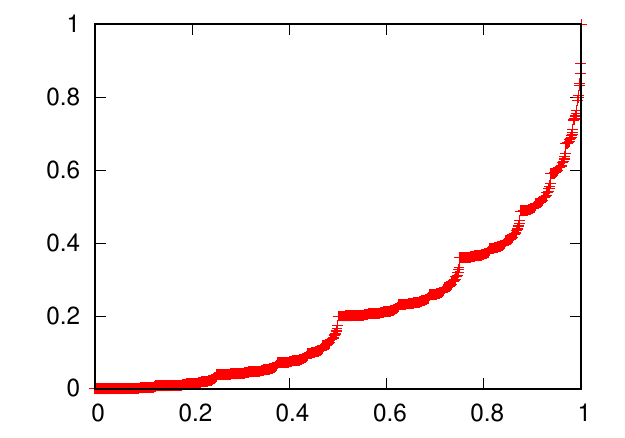}\\
(b)  $\alpha = 1/5$
\end{minipage}
\begin{minipage}[b]{0.32\linewidth}
\centering
\includegraphics[width=1.\linewidth]{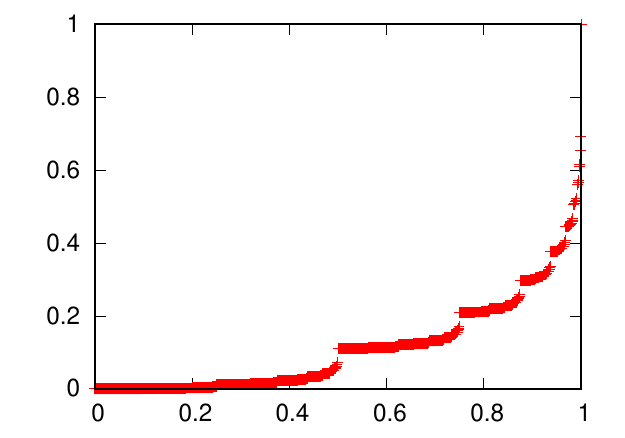}\\
(c)  $\alpha = 1/9$
\end{minipage}
\caption{$L_{\alpha}(x)$ for $x \in [0, 1]$}
\label{fig:Lx}
\end{figure}

\section{Main results}
\label{sec:main}

In this section, we prepare two $D$-dimensional linear symmetric $2$-state radius-$1$ CAs, $F_D$ and $G_D$, and show that for any dimension $D \geq 1$, they are CAs that yield Salem's singular function $L_{1/(2D+1)}$ and $L_{1/(2^D+1)}$, respectively.


Here, we consider the following two $D$-dimensional linear symmetric $2$-state radius-$1$ CAs.

\begin{dfn}
\label{dfn:fdgd}
Let ${\textit{\textbf e}}_{(j)} = (0, \ldots, 0, \underbrace{1}_{\text{$j$th}}, 0, \ldots, 0) \in \{0, 1\}^D$ for $1 \leq j \leq D$, and $k = \sum_{j=1}^D k_j 2^j$ for $0 \leq k < 2^D$.
CAs, $(\{0, 1\}^{{\mathbb Z}^D}, F_D)$ and $(\{0, 1\}^{{\mathbb Z}^D}, G_D)$ are given by
\begin{align}
\label{eq:fd}
(F_D u)_{\textit{\textbf i}} 
&= \sum_{j=1}^D u_{{\textit{\textbf i}} + {\textit{\textbf e}}_{(j)}} + u_{{\textit{\textbf i}} - {\textit{\textbf e}}_{(j)}} \pmod 2, \\
\label{eq:gd}
(G_D u)_{\textit{\textbf i}} 
&= \sum_{k=0}^{2^D-1} u_{{\textit{\textbf i}} + \sum_{j=1}^D (-1)^{k_j}{\textit{\textbf e}}_{(j)}} \pmod 2,
\end{align}
for ${\textit{\textbf i}}=(i_1, \ldots, i_D) \in {\mathbb Z}^D$.
\end{dfn}

The local rule of $F_D$ depends on $2D$ states, and the local rule of $G_D$ depends on $2^D$ states.

\begin{exm}
We give the local rules of $F_D$ and $G_D$ for $D=1$, $2$, and $3$.
When $D=1$, for $u \in \{0, 1\}^{\mathbb Z}$,
\begin{align}
( F_1 u)_i &= ( G_1 u)_i =u_{i-1} + u_{i+1} \pmod 2.
\end{align}
$F_1$ and $G_1$ are Rule $90$ based on Wolfram code \cite{wolfram2002}, and the limit set of its spatio-temporal pattern is well known to be the Sierpinski gasket.
When $D=2$, for $u \in \{0, 1\}^{{\mathbb Z}^2}$,
\begin{align}
( F_2 u)_{i_1, i_2} &= u_{i_1-1, i_2} + u_{i_1, i_2-1} + u_{i_1, i_2+1} + u_{i_1+1, i_2} \pmod 2, \\
( G_2 u)_{i_1, i_2} &= u_{i_1-1, i_2-1} + u_{i_1-1, i_2+1} + u_{i_1+1, i_2-1} + u_{i_1+1, i_2+1} \pmod 2.
\end{align}
The spatio-temporal patterns of $F_2$ and $G_2$ from $u_o$ are shown in Figure~\ref{fig:f2g2}.
For $D=3$ and $u \in \{0, 1\}^{{\mathbb Z}^3}$, we have
\begin{align}
( F_3 u)_{i_1, i_2, i_3} &= u_{i_1+1, i_2, i_3} + u_{i_1-1, i_2, i_3} + u_{i_1, i_2+1, i_3} + u_{i_1, i_2-1, i_3} \nonumber \\
& \qquad  + u_{i_1, i_2, i_3+1} + u_{i_1, i_2, i_3-1} \pmod 2, \\
( G_3 u)_{i_1, i_2, i_3} &= u_{i_1+1, i_2+1, i_3+1} + u_{i_1-1, i_2+1, i_3+1} + u_{i_1+1, i_2-1, i_3+1} \nonumber\\
& \qquad + u_{i_1+1, i_2+1, i_3-1} + u_{i_1-1, i_2-1, i_3+1} + u_{i_1-1, i_2+1, i_3-1} \nonumber\\
& \qquad + u_{i_1+1, i_2-1, i_3-1} + u_{i_1-1, i_2-1, i_3-1} \pmod 2.
\end{align}
\end{exm}

\begin{figure}[h]
\centering
\includegraphics[width=1.\linewidth]{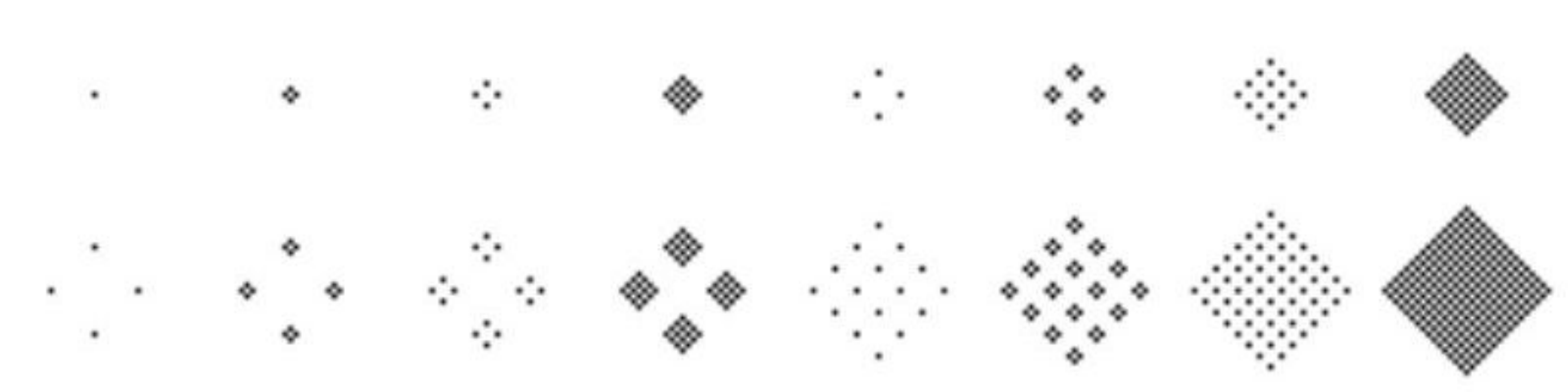}\\
(a) $\{F_2^n u_o\}_{n=0}^{2^4-1}$\\
\includegraphics[width=1.\linewidth]{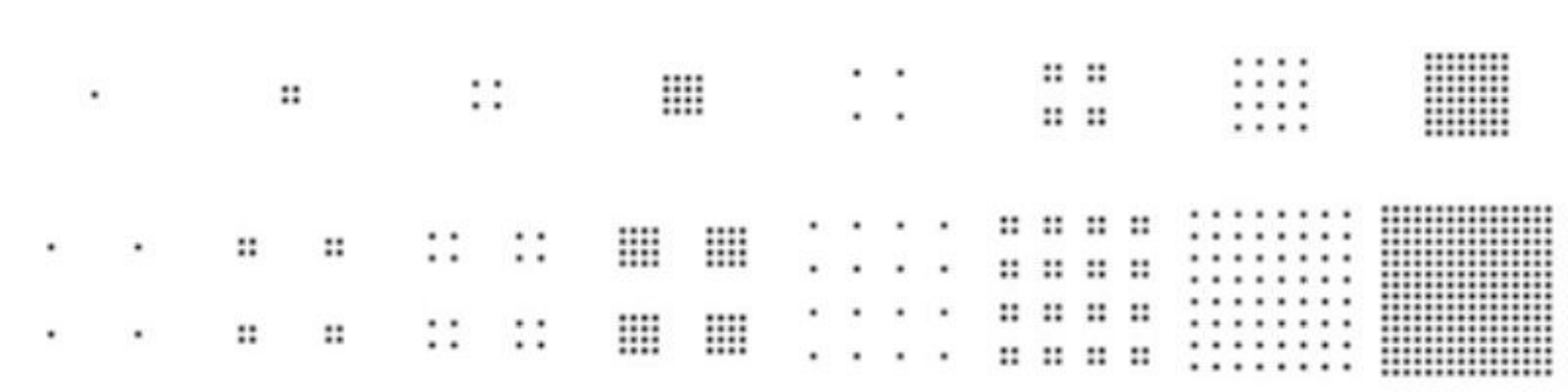}\\
(b)  $\{G_2^n u_o\}_{n=0}^{2^4-1}$
\caption{Spatio-temporal patterns of $F_2$ and $G_2$}
\label{fig:f2g2}
\end{figure}

The following theorem by Takahashi is known as a result for linear CAs.
In the original work \cite{takahashi1992}, the result is given for a $p^k$-state CA where $p$ is a prime number and $k \in {\mathbb Z}_{>0}$.
In the present work, it suffices to consider only the two-state case, which we present here.

\begin{thm}[\cite{takahashi1992}]
\label{thm:tkhs03}
For a linear $2$-state CA, $(T^{2n} u_o)_{2 {\textit{\textbf i}}} = (T^n u_o)_{\textit{\textbf i}}$ for $n \in {\mathbb Z}_{\geq 0}$ and ${\textit{\textbf i}} \in {\mathbb Z}^D$.
If time step $n$ is even and at least one of the elements of ${\textit{\textbf i}}=(i_1, i_1, \ldots, i_D)$ is odd,
then $(T^n u_o)_{\textit{\textbf i}}$ equals $0$.
\end{thm}

\begin{rmk}
\label{rmk:selfsimi}
By Theorem~\ref{thm:tkhs03}, we obtain the self-similarities of $\{ F_D^n u_o\}_{n=0}^{2^k-1}$ and $\{ G_D^n u_o\}_{n=0}^{2^k-1}$.

For the initial configuration $u_o$, and time step $n=1$, we have
\begin{align}
(F_D u_o)_{{\textit{\textbf i}}} &= 
\left\{
\begin{array}{l l}
1 & \ ({\textit{\textbf i}} = \pm {\textit{\textbf e}}_{(j)} \mbox{ for $1 \leq j \leq D$}),\\
0 & \ (otherwise),
\end{array}
\right.\\
(G_D u_o)_{{\textit{\textbf i}}} &= 
\left\{
\begin{array}{l l}
1 & \ ({\textit{\textbf i}} = \sum_{j=1}^D (-1)^{k_j}{\textit{\textbf e}}_{(j)} \mbox{ for $1 \leq k < 2^D$}),\\
0 & \ (otherwise),
\end{array}
\right.
\end{align}
where ${\textit{\textbf e}}_{(j)} = (0, \ldots, 0, \underbrace{1}_{\text{$j$th}}, 0, \ldots, 0) \in \{0, 1\}^D$ for $1 \leq j \leq D$, and $k = \sum_{j=1}^D k_j 2^j$ for $0 \leq k < 2^D$.

By Theorem~\ref{thm:tkhs03}, we obtain the spatio-temporal patterns, $\{F_D^n u_o\}_{n=0}^{2^{k+1}-1}$ and $\{G_D^n u_o\}_{n=0}^{2^{k+1}-1}$ from the patterns, $\{F_D^n u_o\}_{n=0}^{2^k-1}$ and $\{G_D^n u_o\}_{n=0}^{2^k-1}$, respectively.
First, we obtain all states $(F_D^n u_o)_{{\textit{\textbf i}}}$ and $(G_D^n u_o)_{{\textit{\textbf i}}}$ if time step $n$ is even and $n \leq 2^{k+1}-2$. 
Next, we consider the distance between $1$-states for an even time step.
There are no adjacent $1$-states for any time step $n$.
Then, we set their coordinates, ${\textit{\textbf i}}$ and $\hat{\textit{\textbf i}} = ( \hat{i}_1, \hat{i}_2, \ldots, \hat{i}_D) \in {\mathbb Z}^D$ such that $| i_j - \hat{i}_j| \geq 2$ for $1 \leq j \leq D$. 
For a time step $2n$, we have $2 | i_j - \hat{i}_j | \geq  4$, and $1$-states for an even time step are at least four cells apart from each other. 
Each $1$-state multiplies into $2D$ $1$-states in $F_D$ for time step $2n+1$, and multiplies into $2^D$ $1$-states in $G_D$.
Therefore, the spatio-temporal patterns of $F_D$ and $G_D$ are self-similar for time step $n=2^k-1$ for $k \in {\mathbb Z}_{\geq 0}$.
\end{rmk}

\begin{exm}
When $D=1$, the result of Remark~\ref{rmk:selfsimi} is as shown in Figure~\ref{fig:takahashi1}.
All states until time step $n=2^k-1$ are given. 
By Theorem~\ref{thm:tkhs03}, the states with even time step and odd sites are $0$ (Figure~\ref{fig:takahashi1}~(a)), and we have all states for even time steps until $n=2^{k+1}-1$ (Figure~\ref{fig:takahashi1}~(b)).
Because all $1$-states in an even time step are at least four cells apart, the $1$-states multiples into two at the next odd time step (Figure~\ref{fig:takahashi1}~(c)).
Therefore, the spatio-temporal pattern is self-similar for each $2^k-1$ with $k \in {\mathbb Z}_{\geq 0}$.

\begin{figure}[h]
\centering
\includegraphics[width=1.\linewidth]{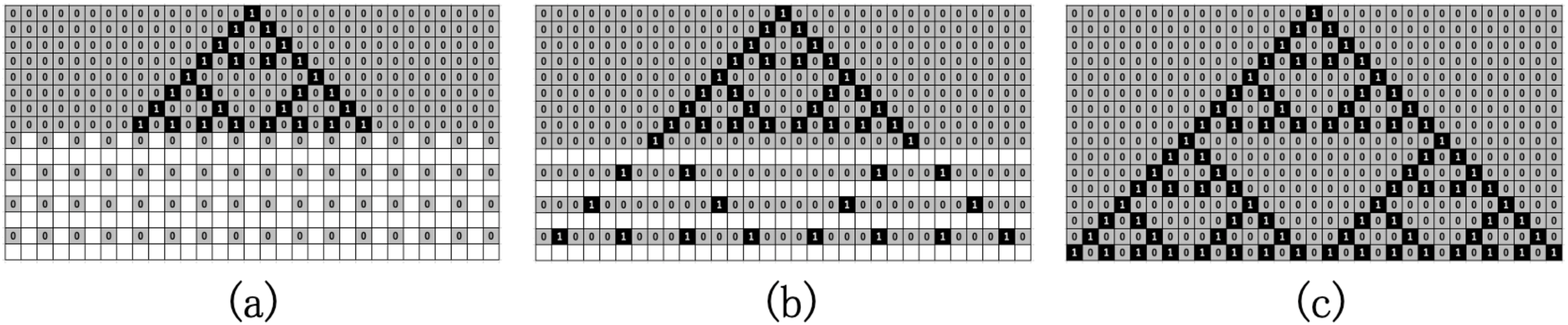}\\
\caption{Constructing the spatio-temporal pattern of $F_1$ (or $G_1$) based on Theorem~\ref{thm:tkhs03}}
\label{fig:takahashi1}
\end{figure}
\end{exm}

\begin{prop}
\label{prop:numcumFG}
For time step $n=2^k-1$, we have 
$cum_{F_D} (2^k-1) = (2D+1)^k$, and $cum_{G_D} (2^k-1) = (2^D+1)^k$.
For time step $n = \sum_{i=0}^{k-1} x_{k-i} 2^i$, we have 
$num_{F_D} (n) = (2D)^{\sum_{i=1}^k x_i}$, and $num_{G_D} (n) = (2^D)^{\sum_{i=1}^k x_i}$.
\end{prop}

\begin{proof}
By Theorem~\ref{thm:tkhs03}, the spatio-temporal patterns of $F_D$ and $G_D$ are self-similar.
The spatio-temporal pattern $\{ F_D^n u_o \}_{n=0}^{2^{k+1}-1}$ consists of $2D+1$ $\{ F_D^n u_o \}_{n=0}^{2^k-1}$, and 
$\{ G_D^n u_o \}_{n=0}^{2^{k+1}-1}$ consists of $2^D+1$ $\{ G_D^n u_o \}_{n=0}^{2^k-1}$, because $num_{F_D}(2^k) = 2D$ and $num_{G_D}(2^k) = 2^D$ (see Figure~\ref{fig:fdgd_selfsimilar}).
Thus, we have $cum_{F_D} (2^k-1) = (2D+1)^k$, and $cum_{G_D} (2^k-1) = (2^D+1)^k$.

Next, we calculate $num_{F_D}(n)$ and $num_{G_D}(n)$.
Because the initial configuration is the single site seed, $num_{F_D}(0) = num_{G_D} = 1$.
By Theorem~\ref{thm:tkhs03}, we have $num_{F_D}(2n) = num_{F_D}(n)$, and $num_{G_D}(2n) = num_{G_D}(n)$, and for odd time steps, $num_{F_D}(2n+1) = 2D \, num_{F_D}(2n)$, and $num_{G_D}(2n+1) = 2^D num_{G_D}(2n)$.
Therefore, $num_{F_D} (n) = (2D)^{\sum_{i=1}^k x_i}$, and $num_{G_D} (n) = (2^D)^{\sum_{i=1}^k x_i}$.
\end{proof}

\begin{figure}[h]
\begin{minipage}[b]{0.47\linewidth}
\centering
\includegraphics[width=1.\linewidth]{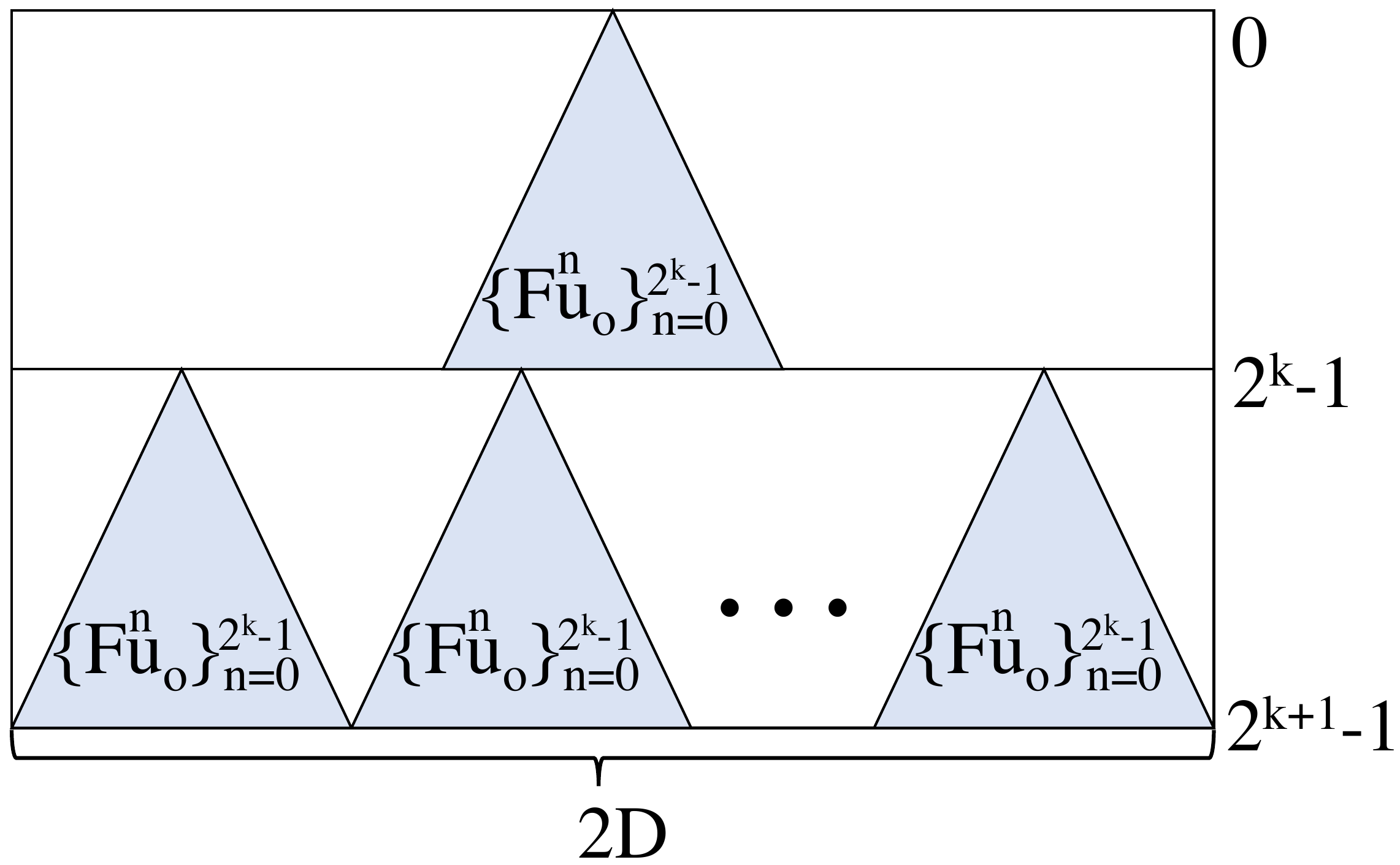}\\
(a) Self-similar sets for $F_D$
\end{minipage}
\begin{minipage}[b]{0.05\linewidth}
\quad
\end{minipage}
\begin{minipage}[b]{0.47\linewidth}
\centering
\includegraphics[width=1.\linewidth]{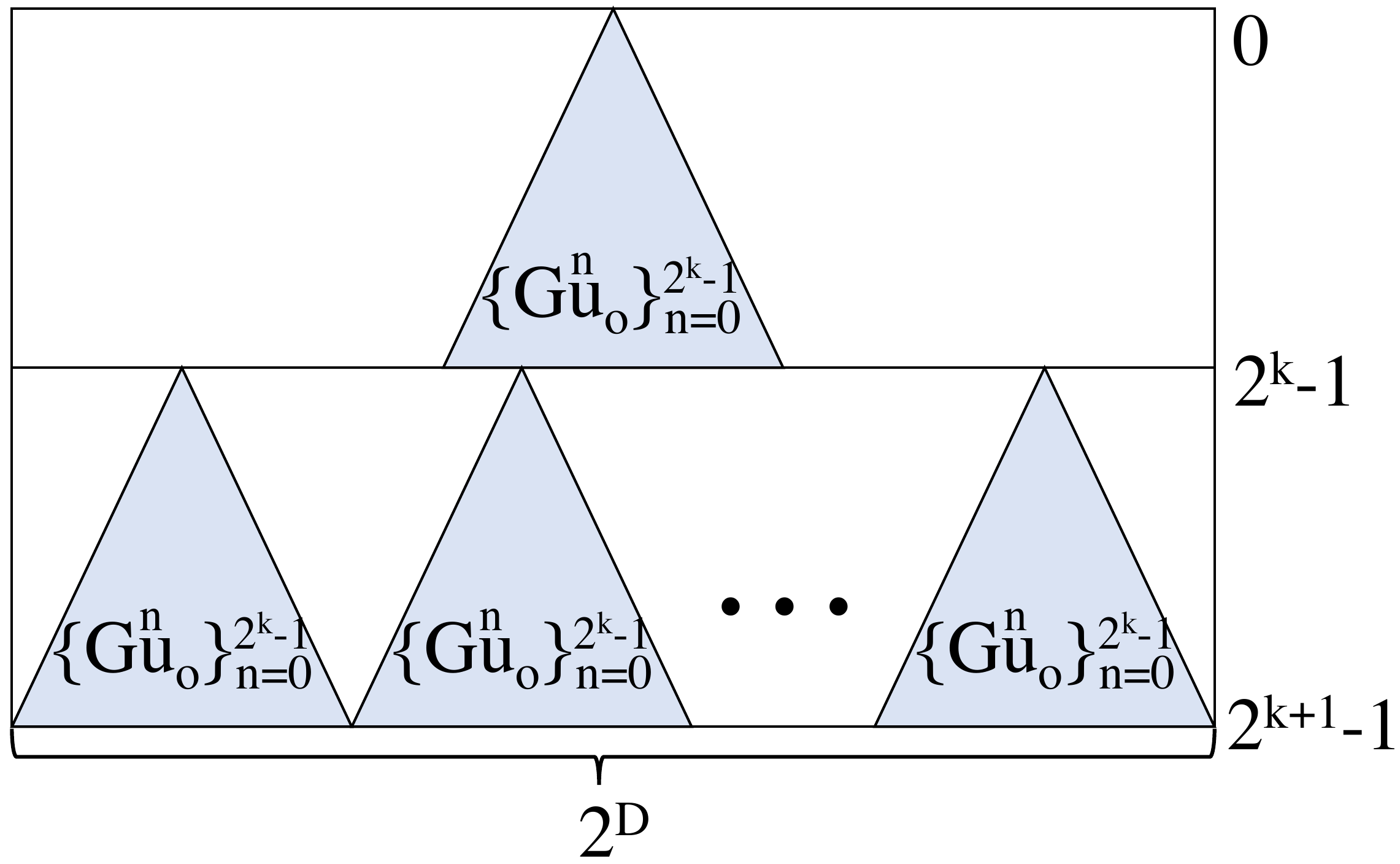}\\
(b) Self-similar sets for $G_D$
\end{minipage}
\caption{Constructing self-similar sets until time step $2^{k+1}-1$ for $F_D$ and $G_D$}
\label{fig:fdgd_selfsimilar}
\end{figure}

\begin{lem}
\label{lem:fin}
Let $n = \sum_{i=0}^{k-1} x_{k-i} \, 2^i \geq 0$, where $x_0 = 0$. 
If a CA $(\{0, 1\}^{{\mathbb Z}^D}, T)$ is $F_D$ or $G_D$, then we have
\begin{align}
\label{eq:fin}
cum_T (n-1) 
&= \sum_{i=1}^k x_i \ num_T \left( \sum_{j=0}^{i-1} x_j 2^{k-j} \right) cum_T (2^{k-i}-1).
\end{align}
\end{lem}

\begin{proof}
By self-similarity of the spatio-temporal patterns, $num_T(2^0) = num_T(2^1)= \cdots = num_T(2^k)$ for $k \in {\mathbb Z}_{\geq}$, and $1$-states at time step $2^k$ are the starting points of the self-similar sets from time step $2^k$ to $2^{k+1}-1$.
Furthermore, $1$-states at time step $\sum_{j=0}^{i-1} x_j 2^{k-j}$ are the starting points of self-similar sets same as the pattern $\{T^n u_o\}_{n=0}^{2^{k-i}-1}$.
Thus, we obtain Equation~\eqref{eq:fin}.
\end{proof}

\begin{exm}
For time step $n=22-1$, we count the nonzero states in the spatio-temporal pattern $\{F_1^n u_o\}_{n=0}^{22-1}$ (Figure~\ref{fig:r90_21}~(a) ).
Figure~\ref{fig:r90_21}~(b) shows the self-similar sets for $\{F_1^n u_o\}_{n=0}^{22-1}$, which comprises one $\{F_1^n u_o\}_{n=0}^{2^4-1}$, two $\{F_1^n u_o\}_{n=0}^{2^2-1}$, and four $\{F_1^n u_o\}_{n=0}^{2^1-1}$.
Based on Equation~\eqref{eq:fin}, we can calculate
\begin{align}
cum_{S90}(22-1) &= num_{S90}(0) \, cum_{S90}(2^4-1) + num_{S90}(2^4) \, cum_{S90} (2^2 -1) \nonumber \\
& \quad + num_{S90}(2^4+2^2) \, cum_{S90} (2^1 -1)\\
&= 1 \cdot 81 + 2 \cdot 9 + 4 \cdot 3 =111.
\end{align}

\begin{figure}[h]
\begin{minipage}[b]{0.45\linewidth}
\centering
\includegraphics[height=25mm]{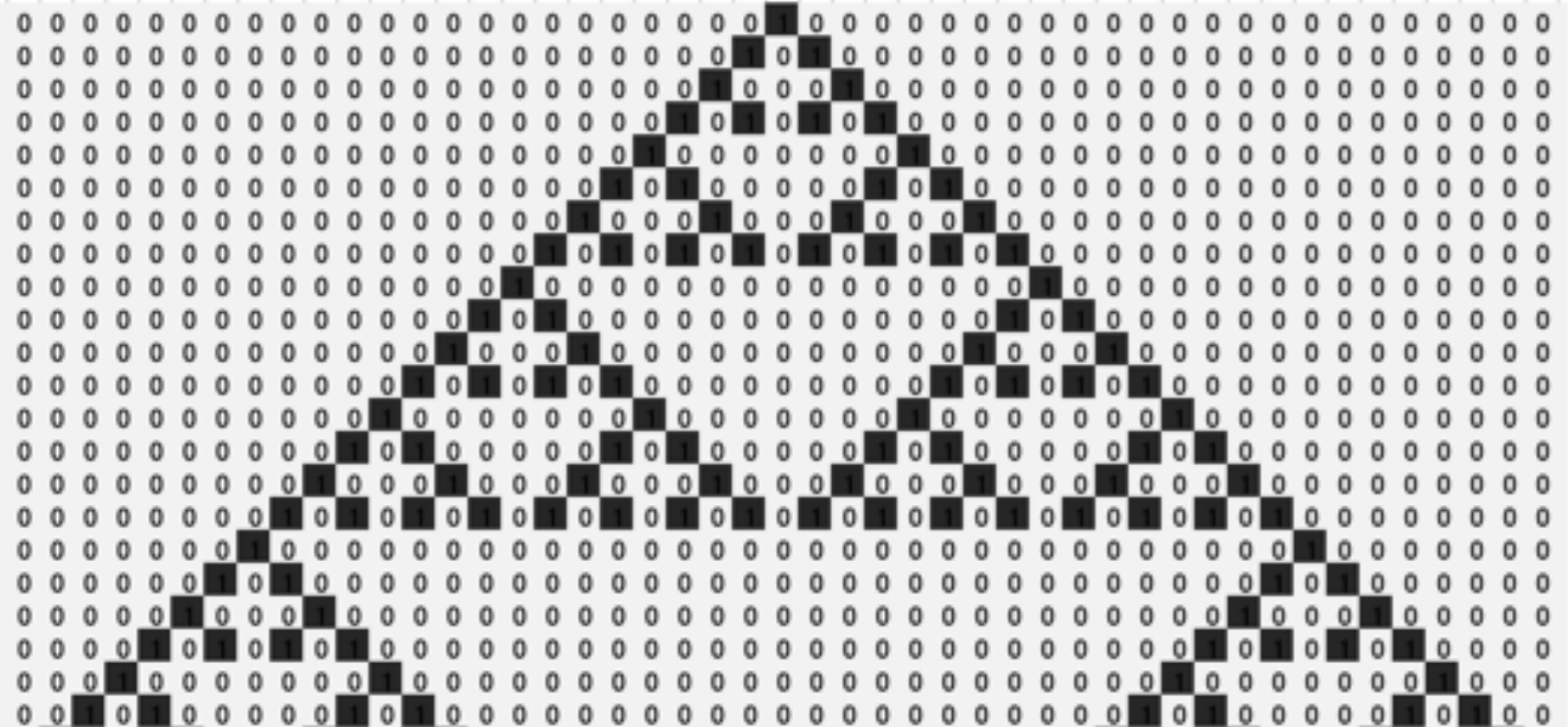}\\
(a) Spatio-temporal patterns $\{F_1^n u_o\}_{n=0}^{22-1}$ ($\{G_1^n u_o\}_{n=0}^{22-1}$)
\end{minipage}
\begin{minipage}[b]{0.05\linewidth}
\quad
\end{minipage}
\begin{minipage}[b]{0.49\linewidth}
\centering
\includegraphics[height=25mm]{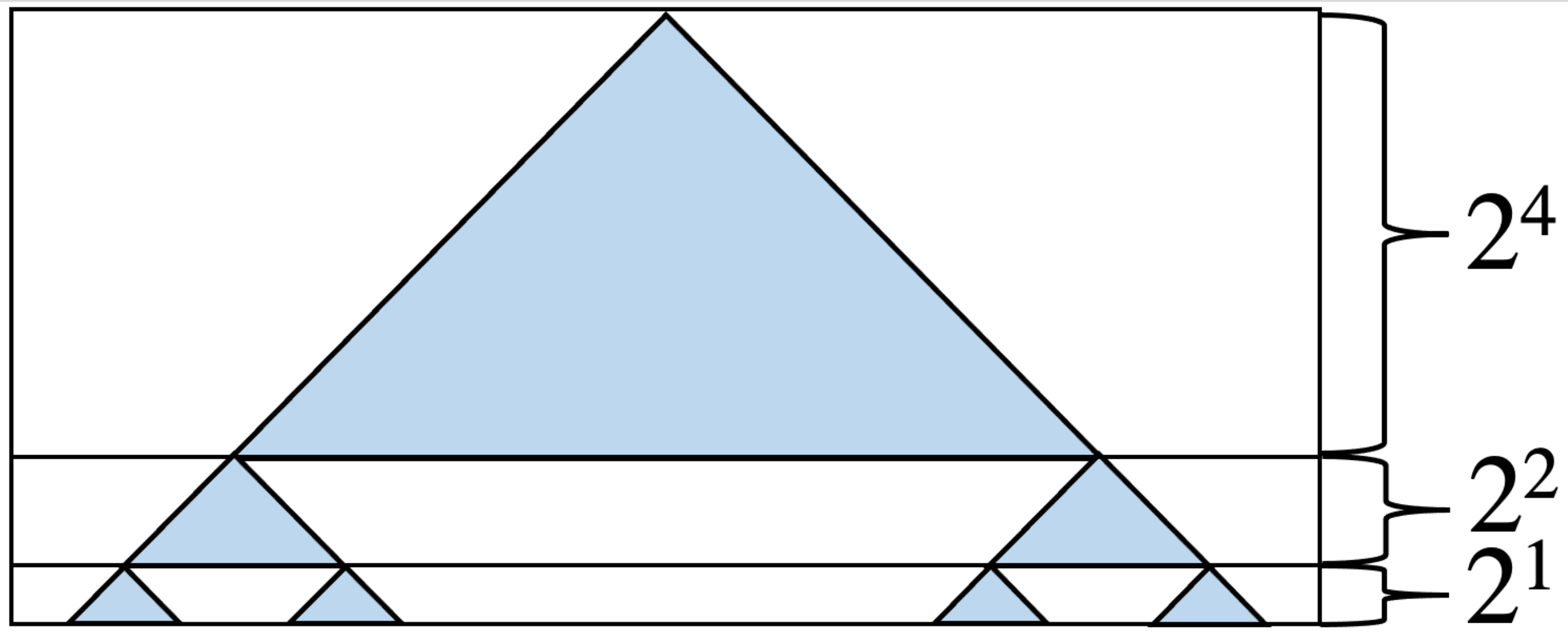}\\
(b) Self-similar sets for $\{F_1^n u_o\}_{n=0}^{22-1}$ ($\{G_1^n u_o\}_{n=0}^{22-1}$)
\end{minipage}
\caption{Counting $cum_{F_1}(n)$ ($cum_{G_1(n)}$) for $n=22-1$}
\label{fig:r90_21}
\end{figure}
\end{exm}

\begin{dfn}
\label{dfn:lim}
For a CA $(\{0, 1\}^{{\mathbb Z}^D}, T)$, we define a function $f_T: [0, 1] \to [0, 1]$ by
\begin{align}
\label{eq:org}
f_T (x) &:= \lim_{k \to \infty} \frac{cum_T\left((\sum_{i=1}^k x_i 2^{k-i})-1 \right)}{cum_T(2^k-1)}
\end{align}
for $x = \sum_{i=1}^{\infty} (x_i/2^i) \in [0,1]$.
\end{dfn}

\begin{thm}
\label{thm:main}
For the CAs, $F_D$ and $G_D$, the functions, $f_{F_D}$ and $f_{G_D}$, exist.
The function $f_{F_D}$ equals Salem's singular function $L_{1/(2D+1)}$, 
and the function $f_{G_D}$ equals Salem's singular function $L_{1/(2^D+1)}$.
\end{thm}

\begin{proof}
By Proposition~\ref{prop:numcumFG} and Lemma~\ref{lem:fin}, we have
\begin{align}
f_{F_D} (x)
&= \lim_{k \to \infty} \frac{cum_{F_D}\left((\sum_{i=1}^k x_i 2^{k-i})-1 \right)}{cum_{F_D}(2^k-1)}\\
&= \lim_{k \to \infty} \frac{\sum_{i=1}^k x_i \ num_{F_D} \left( \sum_{j=1}^{i-1} x_j 2^{k-j} \right) cum_{F_D}(2^{k-i}-1)}{cum_{F_D}(2^k-1)}\\
&= \lim_{k \to \infty} \sum_{i=1}^k x_i \ (2D)^{\sum_{j=1}^{i-1} x_j} (2D+1)^{-i}\\ 
&= \sum_{i=1}^{\infty} x_i \ (2D)^{\sum_{j=1}^{i-1} x_j} (2D+1)^{-i}.
\end{align}
We take $i \in {\mathbb Z}_{>0}$ such that $x_i=1$, and take $k \in {\mathbb Z}_{>0}$ such that $x_{i+k}=1$ and $x_{i+l}=0$ for $0 < l < k$.
By d'Alembert's ratio test, 
we have $|(2D)^{\sum_{j=1}^{i-1+k} x_j} (2D+1)^{-i-k} / ((2D)^{\sum_{j=1}^{i-1} x_j} (2D+1)^{-i})|
\leq (2D/(2D+1))^k < 1$.
Then, the series converges absolutely, and $f_{F_D}$ exists.

We show that $f_{F_D}$ is $L_{1/(2D+1)}$ in Definition~\ref{def:lb}.
Because $D \in {\mathbb Z}_{>0}$, $0 < 1/(2D+1) <1$ and $1/(2D+1) \neq 1/2$.
If $0 \leq x < 1/2$, i.e., $x_1=0$, 
\begin{align}
\frac{1}{2D+1} f_{F_D} (2x) 
&= \frac{1}{2D+1} \sum_{i=1}^{\infty} x_{i+1} \ (2D)^{\sum_{j=1}^{i-1} x_{j+1}} (2D+1)^{-i}\\
&= \frac{1}{2D+1} \sum_{i=2}^{\infty} x_i \ (2D)^{\sum_{j=2}^{i-1} x_j} (2D+1)^{-i+1}\\
&= \sum_{i=1}^{\infty} x_i \ (2D)^{\sum_{j=1}^{i-1} x_j} (2D+1)^{-i}\\
&= f_{F_D} (x).
\end{align}
In contrast, if $1/2 \leq x \leq 1$, when $1= \sum_{i=1}^{\infty} 2^{-i}$, $x_1=1$. 
\begin{align}
&\left( 1- \frac{1}{2D+1} \right) f_{F_D} (2x-1) + \frac{1}{2D+1}\\
&= \frac{1}{2D+1} + \frac{2D}{2D+1} \sum_{i=2}^{\infty} x_i \ (2D)^{\sum_{j=2}^{i-1} x_j} (2D+1)^{-i+1} \\
&= x_1 (2D+1)^{-1} + \sum_{i=2}^{\infty} x_i \ (2D)^{\sum_{j=1}^{i-1} x_j} (2D+1)^{-i} \\
&= \sum_{i=1}^{\infty} x_i \ (2D)^{\sum_{j=1}^{i-1} x_j} (2D+1)^{-i}\\
&= f_{F_D} (x).
\end{align}
Hence, $f_{F_D}=L_{1/(2D+1)}$.

Next, for $G_D$, we have
\begin{align}
f_{G_D} (x)
&= \lim_{k \to \infty} \frac{cum_{G_D}\left((\sum_{i=1}^k x_i 2^{k-i})-1 \right)}{cum_{G_D}(2^k-1)}\\
&= \lim_{k \to \infty} \frac{\sum_{i=1}^k x_i \ num_{G_D} \left( \sum_{j=1}^{i-1} x_j 2^{k-j} \right) cum_{G_D}(2^{k-i}-1)}{cum_{G_D}(2^k-1)}\\
&= \sum_{i=1}^{\infty} x_i \ (2^D)^{\sum_{j=1}^{i-1} x_j} (2^D+1)^{-i}.
\end{align}
We take $i \in {\mathbb Z}_{>0}$ such that $x_i=1$, and take $k \in {\mathbb Z}_{>0}$ such that $x_{i+k}=1$ and $x_{i+l}=0$ for $0 < l < k$.
By d'Alembert's ratio test, we obtain $|(2^D)^{\sum_{j=1}^{i-1+k} x_j} (2^D+1)^{-i-k} / ((2^D)^{\sum_{j=1}^{i-1} x_j} (2^D+1)^{-i})|
\leq (2^D/(2^D+1))^k < 1$.
Then, the series converges absolutely, and $f_{G_D}$ exists.

We show that $f_{G_D}$ is $L_{1/(2^D+1)}$ in Definition~\ref{def:lb}.
Because $D \in {\mathbb Z}_{>0}$, $0 < 1/(2^D+1) <1$ and $1/(2^D+1) \neq 1/2$.
If $0 \leq x < 1/2$, i.e., $x_1=0$, 
\begin{align}
\frac{1}{2^D+1} f_{G_D} (2x) 
&= \frac{1}{2^D+1} \sum_{i=1}^{\infty} x_{i+1} \ (2^D)^{\sum_{j=1}^{i-1} x_{j+1}} (2^D+1)^{-i}
= f_{G_D} (x).
\end{align}
In contrast, if $1/2 \leq x \leq 1$, when $1= \sum_{i=1}^{\infty} 2^{-i}$, $x_1=1$. 
\begin{align}
&\left( 1- \frac{1}{2^D+1} \right) f_{G_D} (2x-1) + \frac{1}{2^D+1}\\
&= \frac{1}{2^D+1} + \frac{2^D}{2^D+1} \sum_{i=2}^{\infty} x_i \ (2^D)^{\sum_{j=2}^{i-1} x_j} (2^D+1)^{-i+1} \\
&= x_1 (2^D+1)^{-1} + \sum_{i=2}^{\infty} x_i \ (2^D)^{\sum_{j=1}^{i-1} x_j} (2^D+1)^{-i} \\
&= f_{G_D} (x).
\end{align}
Therefore, $f_{G_D}=L_{1/(2^D+1)}$.
\end{proof}

\begin{rmk}
We can find the box-counting dimension, a fractal dimension, of the limit set of each spatio-temporal pattern.
For $F_D$, it is $\log(2D+1)/\log 2$, and fot $G_D$, it is $\log(2^D+1)/\log 2$.
\end{rmk}

\section{Numerical results}
\label{sec:other}

In the previous section, we showed that the functions, $f_{F_D}$ and $f_{G_D}$, are the singular functions, $L_{1/(2D+1)}$ and $L_{1/(2^D+1)}$, respectively.
In this section, we consider CAs except $F_D$ and $G_D$, and check whether there exists $T$ the function $f_T$ of which is $L_{1/(M+1)}$ for $M \in {\mathbb Z}_{>1}$.
If a spatio-temporal pattern of a CA is self-similar, the pattern $\{T^n u_o\}_{n=2^k}^{2^{k+1}-1}$ consists of $M$ copies of $\{T^n u_o\}_{n=0}^{2^k-1}$, and then $cum_T (2^k-1) / cum_T (2^{k+1}-1) = 1/(M+1)$ for each $k \in {\mathbb Z}_{> 0}$.
Because $L_{1/(M+1)}(1/2) = 1/(M+1)$, it is possible that if a CA $T$ satisfies $cum_T (2^k-1) / cum_T (2^{k+1}-1) = 1/(M+1)$, the CA may give $f_T = L_{1/(M+1)}$.
To conduct some numerical experiments, we introduce the following function $f_{T, k}$ for a CA $T$ and a finite integer $k$, which is a function $f_T$ before taking the limit of $k$ in Definition~\ref{dfn:lim}. 

\begin{dfn}
\label{dfn:limF}
For a CA $T$, we define a function $f_{T, k}: [0, 1] \to [0, 1]$ by
\begin{align}
\label{eq:orgF}
f_{T, k} (x) &:= \frac{cum_T\left((\sum_{i=1}^k x_i 2^{k-i})-1 \right)}{cum_T(2^k-1)}
\end{align}
for $x = \sum_{i=1}^k (x_i/2^i) \in [0,1]$.
\end{dfn}

Section~\ref{subsec:C_D} provides numerical results on CAs in ${\mathcal C}(D)$ for $D \leq 5$.
We check whether CAs $T$ except $F_D$ and $G_D$ satisfy $f_{T, k} (1/2) = 1/(M+1)$ for each $k \in {\mathbb Z}_{> 0}$. 
Section~\ref{subsec:nonC_D} provides numerical results on linear symmetric $2$-state radius-$1$ CAs on the triangular and hexagonal lattices.
We also check whether there exist CAs the resulting functions of which are $L_{1/(M+1)}$.

\subsection{Numerical results for CAs $\in {\mathcal C}(D) \backslash \{F_D, G_D\}$ for $D \leq 5$}
\label{subsec:C_D}

Recall that the set of $D$-dimensional linear symmetric $2$-state radius-$1$ CAs is ${\mathcal C}(D)$.
A $D$-dimensional linear symmetric $2$-state radius-$1$ CA $(\{0, 1\}^{{\mathbb Z}^D}, T)$ is given by
\begin{align}
(T u)_{\textit{\textbf i}} 
&= \! \! \! \! \sum_{e_1=-1}^1 \sum_{e_2=-1}^1 \! \! \cdots \! \! \sum_{e_D=-1}^1 a_{\sum_{j=1}^D |e_j|} \ u_{i_1+e_1, i_2+e_2, \ldots, i_D+e_D} \! \! \! \! \pmod 2,
\end{align}
for ${\textit{\textbf i}}=(i_1, i_2, \ldots, i_D) \in {\mathbb Z}^D$ and $a_{\sum_{j=1}^D |e_j|} \in \{0, 1\}$ for $-1 \leq e_1, e_2, \ldots, e_D \leq 1$.
By Proposition~\ref{prop:2d+1}, the number of CAs belonging to ${\mathcal C}(D)$ is $2^{D+1}$.

For the one-dimensional case, a CA $(\{0, 1\}^{\mathbb Z}, T) \in {\mathcal C}(1)$ is given by
\begin{align}
(T u)_i &= a_1 u_{i-1} + a_0 u_i + a_1 u_{i+1} \pmod 2,
\end{align}
for $i \in {\mathbb Z}$ and $a_1, a_0 \in \{0, 1\}$.
Four CAs in ${\mathcal C}(1)$ are given by
$(a_1, a_0) = (0, 0), (0, 1), (1, 0), (1, 1)$, and Table~\ref{tab:1dsymCAs} shows the local rules.
By Remark~\ref{rmk:triv}, for $(a_1, a_0)=(0, 0)$ and $(0, 1)$, the orbits are trivial.
For $(a_1, a_0)=(1, 0)$ and $(1, 1)$, the CAs have non-trivial orbits.
The case $(a_1, a_0)=(1, 0)$ provides $F_1$ and $G_1$, that is, Rule $90$.
The case $(a_1, a_0)=(1, 1)$ provides Rule $150$, and the resulting function $f_{S150}$ is known to be a singular function that is not $L_{\alpha}$ \cite{kawa2022}.
Hence, in ${\mathcal C}(1)$ except $F_1$ and $G_1$, there are no CAs the resulting function $f_T$ of which equals $L_{\alpha}$.

\begin{table}[hbtp]
\caption{CAs in ${\mathcal C}(1)$}
\label{tab:1dsymCAs}
\centering
\begin{tabular}{l l l}
\hline
$(a_1, a_0)$ & local rule & Wolfram number\\
\hline \hline
$(0, 0)$ & $0$ & Rule $0$\\
$(0, 1)$ & $u_i$ & Rule $240$\\
$(1, 0)^{\ast 1}$ & $u_{i-1} + u_{i+1} \pmod 2$ & Rule $90$\\
$(1, 1)$ & $u_{i-1} + u_i + u_{i+1} \pmod 2$ & Rule $150$\\
\hline
\end{tabular}\\ \quad \\
(The CA with $^{\ast 1}$ provides $F_1$ and $G_1$.)
\end{table}

Below, we describe our numerical study for CAs in ${\mathcal C}(D)$ for $D=2$, $3$, $4$, and $5$.

\begin{itemize}
\item[$(i)$] 
A CA $(\{0, 1\}^{{\mathbb Z}^2}, T) \in {\mathcal C}(2)$ is given by
\begin{align}
(T u)_{\textit{\textbf i}} 
&= \sum_{e_1=-1}^1 \sum_{e_2=-1}^1 a_{|e_1|+|e_2|} \ u_{i_1+e_1, i_2+e_2} \pmod 2,
\end{align}
for ${\textit{\textbf i}}=(i_1, i_2) \in {\mathbb Z}^2$ and $a_{|e_1|+|e_2|} \in \{0, 1\}$ for $-1 \leq e_1, e_2 \leq 1$.
The number of CAs in ${\mathcal C}(2)$ is $2^{2+1}=8$, and
Table~\ref{tab:2dsymCAs} shows their local rules.
The case $(a_2, a_1, a_0)=(0, 1, 0)$ provides $F_2$, and $(a_2, a_1, a_0)=(1, 0, 0)$ provides $G_2$.
By numerical experiments, for the other six CAs, we obtained $1/ f_{T, k}(1/2) \notin {\mathbb Z}_{>2}$ for $2 \leq k \leq 8$.
We performed experiments only up to $k=8$ owing to computational limitations.

\begin{table}[hbtp]
\caption{CAs in ${\mathcal C}(2)$}
\label{tab:2dsymCAs}
\centering
\begin{tabular}{l l}
\hline
$(a_2, a_1, a_0)$ & local rule $\pmod 2$\\
\hline \hline
$(0, 0, 0)$ & $0$\\
$(0, 0, 1)$ & $u_{i_1, i_2}$\\
$(0, 1, 0)^{\ast 1}$ & $u_{i_1-1, i_2} + u_{i_1, i_2-1} + u_{i_1, i_2+1} + u_{i_1+1, i_2}$\\
$(0, 1, 1)$ & $u_{i_1-1, i_2} + u_{i_1, i_2-1} + u_{i_1, i_2} + u_{i_1, i_2+1} + u_{i_1+1, i_2}$\\
$(1, 0, 0)^{\ast 2}$ & $u_{i_1-1, i_2-1} + u_{i_1-1, i_2+1} + u_{i_1+1, i_2-1} + u_{i_1+1, i_2+1}$\\
$(1, 0, 1)$ & $u_{i_1-1, i_2-1} + u_{i_1-1, i_2+1} + u_{i_1, i_2} + u_{i_1+1, i_2-1} + u_{i_1+1, i_2+1}$\\
$(1, 1, 0)$ & $u_{i_1-1, i_2-1} + u_{i_1-1, i_2+1} + u_{i_1+1, i_2-1} + u_{i_1+1, i_2+1}$\\
&\quad $+ u_{i_1-1, i_2} + u_{i_1, i_2-1} + u_{i_1, i_2+1} + u_{i_1+1, i_2}$\\
$(1, 1, 1)$ & $u_{i_1-1, i_2-1} + u_{i_1-1, i_2+1} + u_{i_1+1, i_2-1} + u_{i_1+1, i_2+1}$\\ 
&\quad $+ u_{i_1-1, i_2} + u_{i_1, i_2-1} + u_{i_1, i_2+1} + u_{i_1+1, i_2} + u_{i_1, i_2}$\\
\hline
\end{tabular}\\ \quad \\
(The CA with $^{\ast 1}$ provides $F_2$, and the CA with $^{\ast 2}$ provides $G_2$.)
\end{table}

\item[$(ii)$]
A CA $(\{0, 1\}^{{\mathbb Z}^3}, T) \in {\mathcal C}(3)$ is given by
\begin{align}
(T u)_{\textit{\textbf i}} 
&= \sum_{e_1=-1}^1 \sum_{e_2=-1}^1 \sum_{e_3=-1}^1 a_{|e_1|+|e_2|+ |e_3|} \ u_{i_1+e_1, i_2+e_2, i_3+e_3} \pmod 2,
\end{align}
for ${\textit{\textbf i}}=(i_1, i_2, i_3) \in {\mathbb Z}^3$ and $a_{|e_1|+|e_2|+|e_3|} \in \{0, 1\}$ for $-1 \leq e_1, e_2, e_3 \leq 1$.
Thus, there exist $2^4$ CAs in ${\mathcal C}(3)$.
A CA given by $(a_3, a_2, a_1, a_0)=(0, 0, 1, 0)$ is $F_3$, and a CA given by $(a_3, a_2, a_1, a_0)=(1, 0, 0, 0)$ is $G_3$. 
By numerical experiments, for the other fourteen CAs, we obtained $1/f_{T, k}(1/2) \notin {\mathbb Z}_{>2}$ for $2 \leq k \leq 6$.
Owing to computational limitations, experiments were performed only up to $k=6$.

\item[$(iii)$]
A CA $(\{0, 1\}^{{\mathbb Z}^4}, T) \in {\mathcal C}(4)$ is given by
\begin{align}
(T u)_{\textit{\textbf i}} 
&= \sum_{e_1=-1}^1 \sum_{e_2=-1}^1 \sum_{e_3=-1}^1 \sum_{e_4=-1}^1 a_{\sum_{j=1}^4 |e_j|} \ u_{i_1+e_1, i_2+e_2, i_3+e_3, i_4+e_4} \pmod 2,
\end{align}
for ${\textit{\textbf i}}=(i_1, i_2, i_3, i_4) \in {\mathbb Z}^4$ and $a_{\sum_{j=1}^4 |e_j|} \in \{0, 1\}$ for $-1 \leq e_1, e_2, e_3, e_4 \leq 1$.
Thus, there exist $2^5$ CAs in ${\mathcal C}(4)$.
A CA given by $(a_4, a_3, a_2, a_1, a_0)=(0, 0, 0, 1, 0)$ is $F_4$, and a CA given by $(a_4, a_3, a_2, a_1, a_0)=(1, 0, 0, 0, 0)$ is $G_4$. 
By numerical experiments, we found that the other $30$ CAs satisfy
$1/f_{T, k}(1/2) \notin {\mathbb Z}_{>2}$ for $k = 3$ and $4$.
Owing to computational limitations, we performed experiments only up to $k=4$.
In contrast to the other dimensional cases, the result does not hold only in the case that $D=4$ for $k=2$.

\item[$(iv)$]
A CA $(\{0, 1\}^{{\mathbb Z}^5}, T) \in {\mathcal C}(5)$ is given by
\begin{align}
(T u)_{\textit{\textbf i}} 
&= \!\!\!\! \sum_{e_1=-1}^1 \sum_{e_2=-1}^1 \!\! \cdots \!\! \sum_{e_5=-1}^1 a_{\sum_{j=1}^5 |e_j|} \ u_{i_1+e_1, i_2+e_2, i_3+e_3, i_4+e_4, i_5+e_5} \pmod 2,
\end{align}
for ${\textit{\textbf i}}=(i_1, i_2, i_3, i_4, i_5) \in {\mathbb Z}^5$ and $a_{\sum_{j=1}^5 |e_j|} \in \{0, 1\}$ for $-1 \leq e_1, e_2, e_3, e_4, e_5 \leq 1$.
There exist $2^6$ CAs in ${\mathcal C}(5)$.
A CA given by $(a_5, a_4, a_3, a_2, a_1, a_0)=(0, 0, 0, 0, 1, 0)$ is $F_5$, and a CA given by $(a_5, a_4, a_3, a_2, a_1, a_0)=(1, 0, 0, 0, 0, 0)$ is $G_5$. 
By numerical experiments, for the other $62$ CAs hold
$1/f_{T, k}(1/2) \notin {\mathbb Z}_{>2}$ for $k = 2$ and $3$.
Owing to computational limitations, experiments were performed only up to $k=3$.
\end{itemize}

Therefore, a CA $T \in {\mathcal C}(D) \backslash \{F_D, G_D\}$ satisfies that $f_{T, k} (x) \neq L_{1/(M+1)}(x)$ for $x = \sum_{i=1}^k (x_i/2^i)$ and any $M \in {\mathbb Z}_{>1}$, if $D=2$ and $2 \leq k \leq 8$, if $D=3$ and $2 \leq k \leq 6$, if $D=4$ and $3 \leq k \leq 4$, and if $D=5$ and $2 \leq k \leq 3$. 
Although this result implies the possibility of satisfying $f_T \neq L_{1/(M+1)}$ for any $M \in {\mathbb Z}_{>1}$, this is not guaranteed, because $f_{T, k}$ could converge to $f_T$ when taking the limit of $k$.

\subsection{Numerical results for CAs on the triangular and hexagonal lattices}
\label{subsec:nonC_D}

It has been established that there exist only three regular tessellations of the plane, the triangular lattice (Figure~\ref{fig:hexa_lattice}~(a)), the square lattice, and the hexagonal lattice (Figure~\ref{fig:hexa_lattice}~(b)).
In this section, we consider linear symmetric $2$-state radius-$1$ CAs on the triangular and hexagonal lattices.

\begin{figure}[h]
\begin{minipage}[b]{0.45\linewidth}
\centering
\includegraphics[width=.65\linewidth]{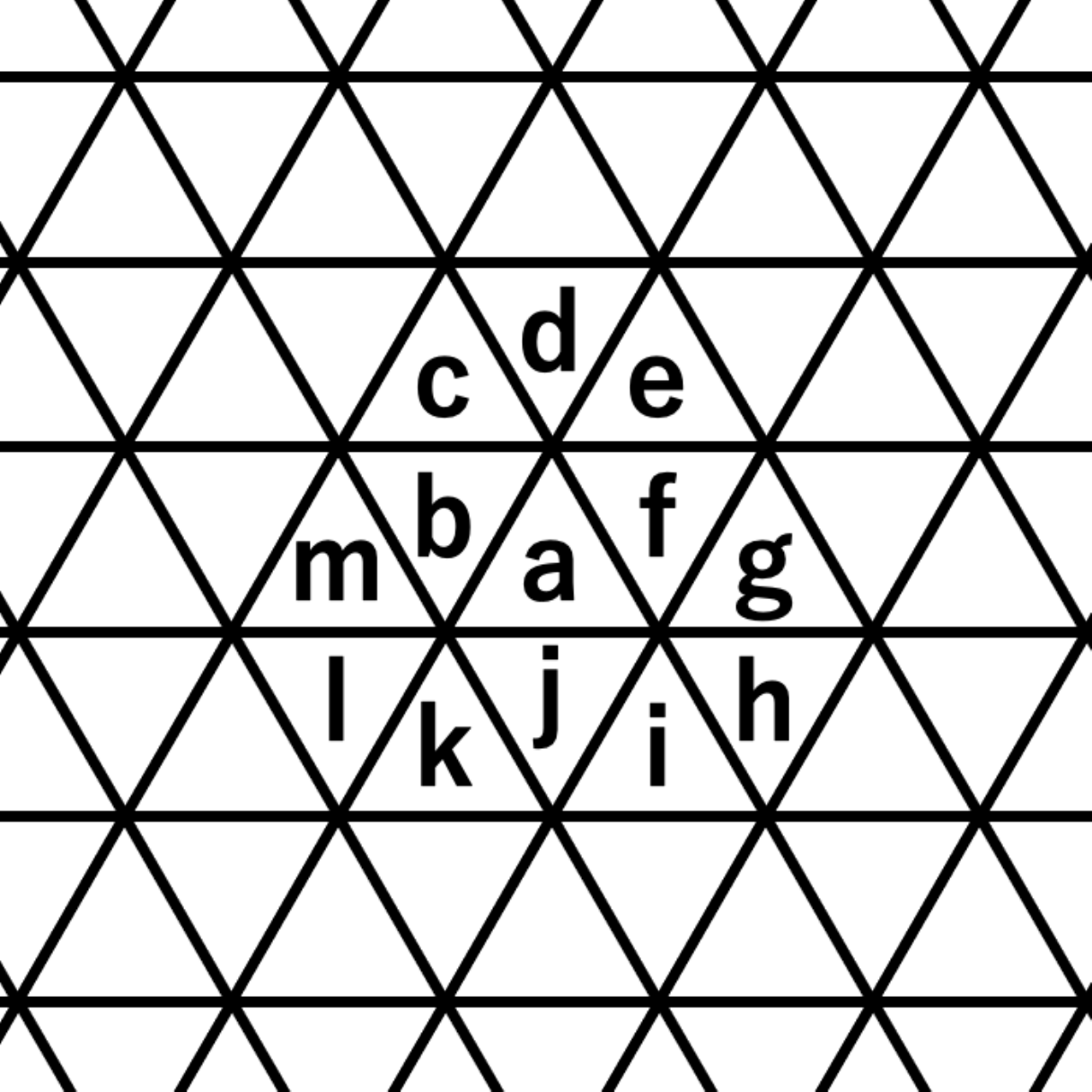}\\
(a) The triangular lattice
\end{minipage}
\begin{minipage}[b]{0.05\linewidth}
\quad
\end{minipage}
\begin{minipage}[b]{0.45\linewidth}
\centering
\includegraphics[width=.65\linewidth]{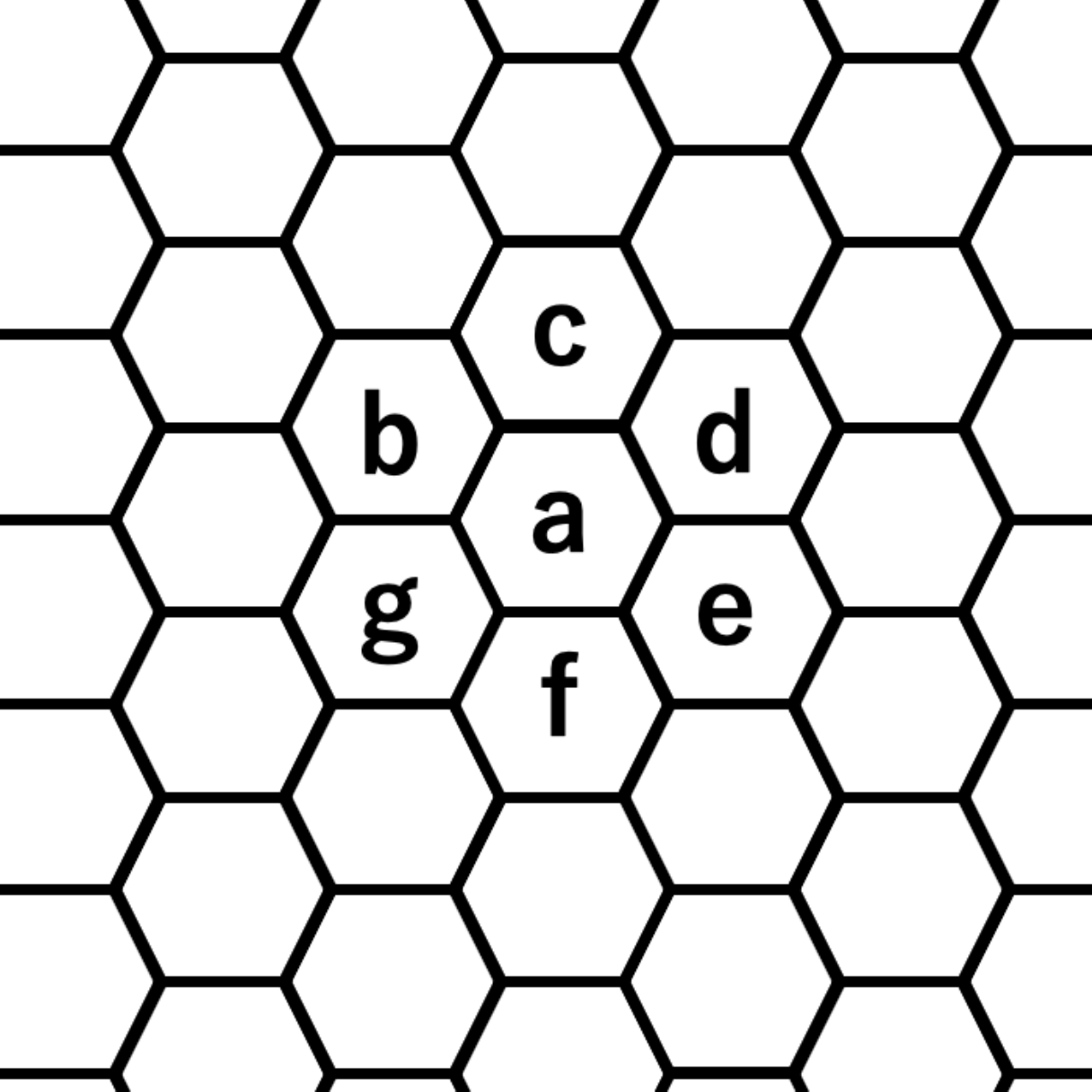}\\
(b) The hexagonal lattice
\end{minipage}
\caption{Two of the three regular lattices that can tessellate the planes, excepting the square lattice}
\label{fig:hexa_lattice}
\end{figure}

On the triangular lattice, local rules of linear symmetric $2$-state radius-$1$ CAs depend on thirteen states of neighbor cells (in Figure~\ref{fig:hexa_lattice}~(a), the cells are given from $a$ to $m$), and Table~\ref{tab:triCAs} shows the fourteen triangular CAs, from $R_0$ to $R_{13}$. 
By numerical experiments, we found that for the fourteen CAs, $T$ satisfies $1/f_{T, k}(1/2) \notin {\mathbb Z}_{>2}$ for $3 \leq k \leq 7$.
Owing to computational limitations, experiments were performed only up to $k=7$.
Hence, a linear symmetric $2$-state radius-$1$ triangular CA satisfies that $f_{T, k} (x) \neq L_{1/(M+1)}(x)$ for $3 \leq k \leq 7$, $x = \sum_{i=1}^k (x_i/2^i)$, and any $M \in {\mathbb Z}_{>1}$. 

\begin{table}[hbtp]
\caption{Linear symmetric $2$-state radius-$1$ triangular CAs}
\label{tab:triCAs}
\centering
\begin{tabular}{c l}
\hline
CA & local rule $\pmod 2$\\
\hline 
$R_0$ & $0$\\
$R_1$ & $a$\\
$R_2$ & $b + f + j$\\
$R_3$ & $a + b + f + j$\\
$R_4$ & $d + h + l$\\
$R_5$ & $a + d + h + l$\\
$R_6$ & $c + e + g + i + k + m$\\
$R_7$ & $a + c + e + g + i + k + m$\\
$R_8$ & $(b + f + j) + (c + e + g + i + k + m)$\\
$R_9$ & $a + (b + f + j) + (c + e + g + i + k + m)$\\
$R_{10}$ & $(d + h + l) + (c + e + g + i + k + m)$\\
$R_{11}$ & $a + (d + h + l) + (c + e + g + i + k + m)$\\
$R_{12}$ & $(b + f + j) + (d + h + l) + (c + e + g + i + k + m)$\\
$R_{13}$ & $a + (b + f + j) + (d + h + l) + (c + e + g + i + k + m)$\\
\hline
\end{tabular}
\end{table}

On the hexagonal lattice, local rules of linear symmetric $2$-state radius-$1$ CAs depend on seven states of neighbor cells (in Figure~\ref{fig:hexa_lattice}~(b), the cells are given from $a$ to $g$), and Table~\ref{tab:hexaCAs} shows the six hexagonal CAs, from $H_0$ to $H_5$. 

The results of numerical experiments showed that if $T=H_0$, $H_1$, $H_4$, or $H_5$, we have $1/f_{T, k}(1/2) \notin {\mathbb Z}_{>2}$ for $3 \leq k \leq 8$.
Owing to computational limitations, experiments were performed only up to $k=8$. 
In contrast, for $H_2$ and $H_3$, we obtained the following results.

\begin{itemize}
\item[$(i)$] 
Figure~\ref{fig:h2h3}~(a) shows the spatio-temporal pattern of $H_2$ from time step $n=144$ to $255$ every $16$ steps, and Figure~\ref{fig:hexa_numbers}~(a) gives the dynamics of the number of nonzero states for time step $n=0$ to $2^8-1$.
We obtain $num_{H2}(n) = 4^{\sum_{i=1}^k x_i}$, and $cum_{H2}(2^k-1) = 5^{\sum_{i=1}^k x_i}$ for time step $n = \sum_{i=0}^{k-1} x_{k-i} 2^i$. 
Hence, we have $f_{H2} = L_{1/5}$.
\item[$(ii)$] 
Figure~\ref{fig:h2h3}~(b) shows the spatio-temporal pattern of $H_3$ from time step $n=144$ to $255$ every $16$ steps, and it may be observed that the spatial pattern for time step $n=2^k-1$ for $k \in {\mathbb Z}_{\geq 0}$ is similar to the spatio-temporal pattern of the one-dimensional CA Rule $90$. 
Figure~\ref{fig:hexa_numbers}~(b) gives the dynamics of the number of nonzero states for time step $n=0$ to $2^8-1$.
We obtain $num_{H3}(n) = 3^{\sum_{i=1}^k x_i}$, and $cum_{H3}(2^k-1) = 4^{\sum_{i=1}^k x_i}$. 
Thus, we have $f_{H3} = L_{1/4}$.
\end{itemize}

\begin{table}[hbtp]
\caption{Linear symmetric $2$-state radius-$1$ hexagonal CAs}
\label{tab:hexaCAs}
\centering
\begin{tabular}{c l}
\hline
CA & local rule $\pmod 2$\\
\hline 
$H_0$ & $0$\\
$H_1$ & $a$\\
$H_2$ & $b + d + f$\\
$H_3$ & $a + b + d + f$\\
$H_4$ & $b+c+d+e+f+g$\\
$H_5$ & $a+b+c+d+e+f+g$\\
\hline
\end{tabular}
\end{table}

\begin{figure}[h]
\centering
\includegraphics[width=1.\linewidth]{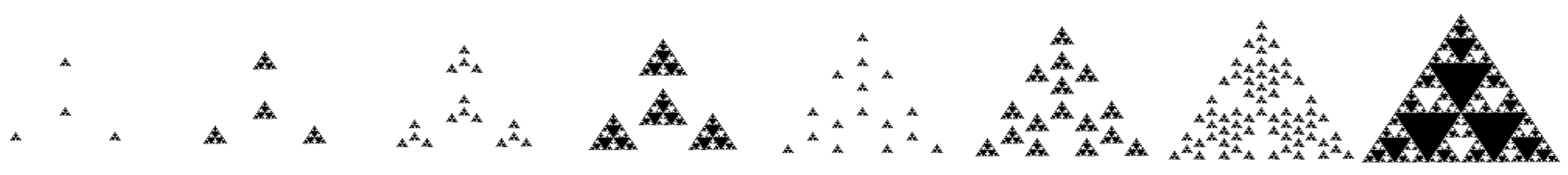}\\
(a) $\{H_2^{2^7+2^4m-1} u_o\}_{m=1}^{2^3}$\\
\includegraphics[width=1.\linewidth]{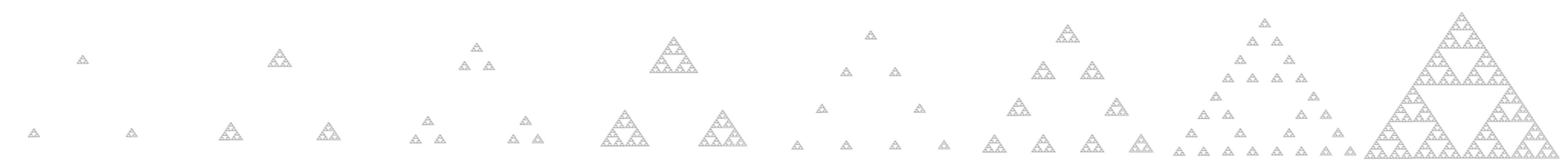}\\
(b)  $\{H_3^{2^7+2^4m-1} u_o\}_{m=1}^{2^3}$
\caption{Spatio-temporal patterns of $H_2$ and $H_3$}
\label{fig:h2h3}
\end{figure}

\begin{figure}[h]
\begin{minipage}[b]{0.45\linewidth}
\centering
\includegraphics[width=1.\linewidth]{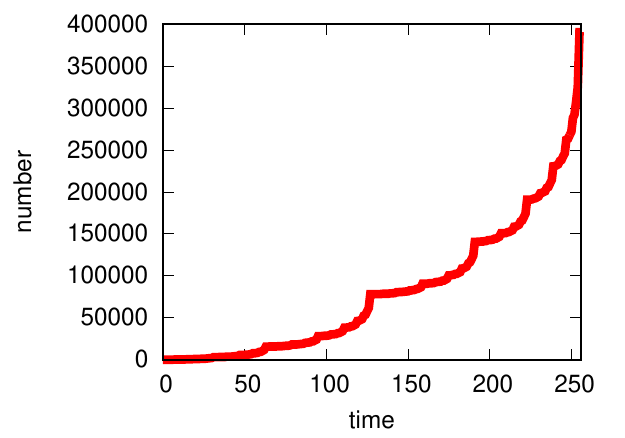}\\
(a) $\{cum_{H2} (n)\}_{n=0}^{2^8-1}$
\end{minipage}
\begin{minipage}[b]{0.05\linewidth}
\quad
\end{minipage}
\begin{minipage}[b]{0.45\linewidth}
\centering
\includegraphics[width=1.\linewidth]{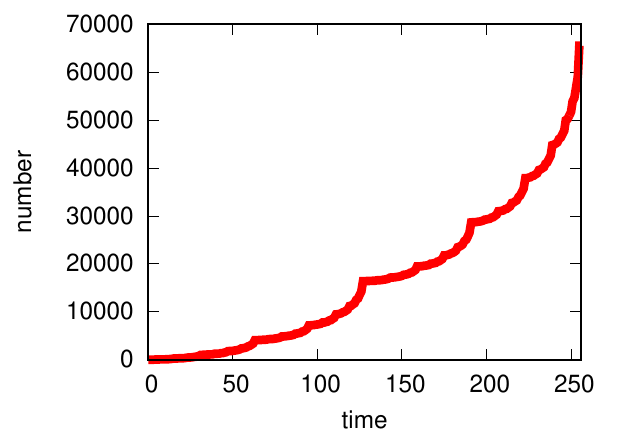}\\
(b) $\{cum_{H3} (n)\}_{n=0}^{2^8-1}$
\end{minipage}
\caption{Dynamics of the number of nonzero states for hexagonal CAs, $H_2$ and $H_3$}
\label{fig:hexa_numbers}
\end{figure}

\section{Concluding remarks}
\label{sec:cr}

In this paper, we have provided functions obtained by the number of nonzero states of the spatio-temporal patterns of $F_D$ and $G_D$ for $D \in {\mathbb Z}_{> 0}$, and we have shown that these functions equal Salem's singular function $L_{\alpha}$ with $\alpha = 1/(2D+1)$ and $1/(2^D+1)$, respectively.
This result includes our previous results for the one-dimensional elementary CA Rule $90$ the function of which is $L_{1/3}$, and a two-dimensional elementary CA with the function $L_{1/5}$. 
Also, this result indicates that there exist CAs the resulting function of which is $L_{1/M}$ for any odd number $M$ greater than or equal to $3$. 
In Section~\ref{sec:other}, we have provided numerical results. 
For $2 \leq D \leq 5$, there are no CAs $T \in {\mathcal C}(D) \backslash \{F_D, G_D\}$ satisfying $f_{T, k} (x) = L_{1/M}(x)$ for $x = \sum_{i=1}^k (x_i/2^i)$ and any $M \in {\mathbb Z}_{>2}$.
We also studied CAs on the triangular and hexagonal lattices.
On the triangular lattice, there are no CAs related to Salem's singular function.
On the hexagonal lattice, we found two CAs related to $L_{1/5}$ and $L_{1/4}$.

In future work, we plan to study a CA $T$ the function $f_T$ of which equals $L_{1/M}$ for an even number $M$.
In this paper, we have shown that we have a hexagonal CA the function of which is $L_{1/4}$ when $M=4$.
However, it remains unclear whether there exists a CA whose resulting function is $L_{1/M}$ for every even number $M$ greater than $4$.
(Given that if $M=2$, the condition $\alpha \neq 1/2$ is not satisfied, the smallest even number $M$ is $4$.)
We also plan to study other CAs.

\section*{Acknowledgment}
This work was partly supported by a Grant-in-Aid for Scientific Research (18K13457, 22K03435) funded by the Japan Society for the Promotion of Science.

\section*{Data Availability Statement}
The data that supports the findings of this work are available within this paper.


\end{document}